%% file: paper.tex
\documentclass[11pt,reqno]{article}
\usepackage[margin=1in]{geometry}
\usepackage{xspace}
\usepackage{xcolor}
\usepackage{amsmath,amsfonts,amsthm,amssymb}
\usepackage{mathtools}
\usepackage{amsfonts}
\usepackage{amssymb}
\usepackage{verbatim}
\usepackage{mathrsfs} 
\usepackage{hyperref} 
\usepackage[toc,page]{appendix}
\usepackage{cleveref}
\usepackage{titlesec}
\titleformat{\subsubsection}[runin]
  {\normalfont\normalsize\bfseries}
  {\thesubsubsection}
  {1em}
  {}
\newtheorem{theorem}{Theorem}[section]
\newtheorem{corollary}[theorem]{Corollary}
\newtheorem{lemma}[theorem]{Lemma}
\newtheorem{proposition}[theorem]{Proposition}
\theoremstyle{definition}
\newtheorem{definition}[theorem]{Definition}

\newtheorem{remark}{\it{\textbf{Remark}}}

\newcommand{\poly}{\mathsf{poly}}

\title{Random Shortening of Linear Codes and Applications}

 \author{Xue Chen\thanks{School of Computer Science, University of Science and Technology of China, Email: xuechencs@gmail.com}
 \and Kuan Cheng\thanks{Center on Frontiers of Computing Studies, Peking University, Email: ckkcdh@pku.edu.cn}
 \and Xin Li\thanks{Department of Computer Science,  Johns Hopkins University, Email: lixints@cs.jhu.edu. Supported by NSF CAREER Award CCF-1845349 and NSF Award CCF-2127575.}
 \and Songtao Mao\thanks{Department of Computer Science, Johns Hopkins University, Email: smao13@jhu.edu. Supported by NSF Award CCF-2127575.}
 }
\date{}

\begin{document}

\maketitle

\input{abstract}

\input{intro}
\input{prelim}

\input{result.tex}

\bibliographystyle{alpha}
\bibliography{references}

\end{document}

%% file: abstract.tex
\begin{abstract}
Random linear codes (RLCs) are well known to have nice combinatorial properties and near-optimal parameters in many different settings. However, getting explicit constructions matching the parameters of RLCs is challenging, and RLCs are hard to decode efficiently. This motivated several previous works to study the problem of partially derandomizing RLCs, by applying certain operations to an explicit mother code. Among them, one of the most well studied operations is \emph{random puncturing}, where a series of works culminated in the work of  Guruswami and Mosheiff (FOCS' 22), which showed that a random puncturing of a low-biased code is likely to possess almost all interesting local properties of RLCs. 

In this work, we provide an in-depth study of another, dual operation of random puncturing, known as \emph{random shortening}, which can be viewed equivalently as random puncturing on the dual code. Our main results show that for any small $\varepsilon$, by starting from a mother code with certain weaker conditions (e.g., having a large distance) and performing a random (or even pseudorandom) shortening, the new code is $\varepsilon$-biased with high probability. Our results hold for any field size and yield a shortened code with constant rate. This can be viewed as a complement to random puncturing, and together, we can obtain codes with properties like RLCs from weaker initial conditions.

Our proofs involve several non-trivial methods of estimating the weight distribution of codewords, which may be of independent interest.
\end{abstract}

%% file: intro.tex
\section{Introduction}
Error correcting codes are fundamental objects in combinatorics and computer science. The study of these objects together with the bounds and parameters that can be achieved, has also helped shape the field of information theory starting from the pioneering work of Shannon and Hamming. In the theory of error-correcting codes, linear codes form a fundamental class of codes that are building blocks of many important constructions and applications. Such codes have simple algebraic structures that are often key ingredients in their performance and analysis. For example, any linear code with message length $k$ and codeword length $n$ over the field $\mathbb{F}_q$ can be described by both a generator matrix in $\mathbb{F}_q^{k \times n}$ and a parity check matrix in $\mathbb{F}_q^{n \times (n-k)}$.

It is well known that random linear codes (RLCs, where one samples each entry of the generator matrix uniformly independently from $\mathbb{F}_q$) enjoy nice combinatorial properties and have near-optimal parameters in many different settings. Specifically, with high probability they achieve Shannon  capacity, the Gilbert-Varshamov (GV) bound of rate-distance tradeoff, and are list-decodable up to capacity. However, getting explicit constructions remains a challenging problem in many situations. In addition, random linear codes have little structure, which makes it difficult to design efficient decoding algorithms. Indeed, decoding random linear codes is closely related to the problems of learning parity with noise and learning with errors, whose hardness is the basis of many cryptographic applications (see e.g., \cite{DBLP:journals/jacm/Regev09}).\ As such, many previous works studied the problem of slightly derandomizing, or equivalently reducing the randomness used in RLCs, while still maintaining their nice properties.


Among these works, random puncturing is one of the most well-studied operations. Here, one takes an explicit mother code, and then randomly punctures some coordinates from the code (or equivalently, punctures some columns from the generator matrix) to get a new, shorter code. Specifically, a $\mathcal{P}$-puncturing of a mother code $\mathcal{C}  \subseteq \mathbb{F}_q^n$   randomly chooses a subset $\mathcal{P}\subseteq [n]$ of size $p$, and for every codeword of $\mathcal{C}$, deletes all symbols with positions in $\mathcal{P}$.
Compared to standard RLCs, the number of random bits used is thus reduced from $O(nk\log q)$ to $O(n)$. Furthermore, certain nice structures of the mother code are often inherited by the punctured code, which makes decoding easier. 

With sophisticated techniques, previous works have shown that if the mother code satisfies some natural conditions, then after a random puncturing, with high probability the new code has certain properties similar to those of RLCs. For example, motivated by the problem of achieving list-decoding capacity, recent works \cite{wootters2013list,rudra2014every,ferber2022list,goldberg2021list,brakensiek2022generic,guo2023randomly,alrabiah2023randomly} studied random puncturing of Reed-Muller (RM) codes and Reed-Solomon (RS) codes. Subsequent works \cite{guruswami2022punctured,pyne2023pseudorandom} generalized the list-decoding property to all monotone-decreasing local properties. In all these works, the mother code needs to have some special properties, such as being an RS code, an RM code, having a large distance over a large alphabet, or having a low bias over a small alphabet. These properties are not immediately implied by general linear codes, and thus, one of the natural goals is to gradually weaken the requirements of the mother code so that the approach works for a broader class of codes. Indeed, as we shall see later, this is one of the main motivations and themes in previous works.  

In this paper we continue this line of work and study the following two natural questions:
\begin{enumerate}
    \item \emph{If the mother code is not that strong, can we still use some operations to get a new code that has properties similar to random linear codes?}
    \item \emph{What other operations, besides random puncturing, are useful in this context?}
\end{enumerate}

Towards answering these questions, we consider a different operation to reduce the randomness of RLCs, called random shortening, previously studied in \cite{bioglio2017low,liu2021shortened,yardi2017shortened,nelson2015existence}.
Specifically, for an integer $s$, a random $s$-shortening of a code $\mathcal{C} \subseteq \mathbb{F}_q^n$ randomly chooses a subset $\mathcal{S} \subseteq [n]$ of size $s$, and forms a new code by picking all codewords of $\mathcal{C}$ which are zeros at the positions in $\mathcal{S}$, and deleting these zero symbols.

We note that just like random puncturing, the operation of random shortening can in fact be carried out on any code, not just on linear codes. However, for linear codes there is an important, alternative view of random shortening: it is actually the dual version of random puncturing. In particular, one can check that it is equivalent to a random puncturing of size $s$ on the parity check matrix of a linear code $C$, or the generator matrix of the dual code $\mathcal{C}^{\perp}$. Thus in this paper, for a linear code, we also call shortening \emph{dual puncturing}.

This view brings some convenience from the viewpoint of the parity check matrix. For example, any puncturing of the parity check matrix (hence also shortening) of a low-density parity check (LDPC) code \cite{gallager1962low} still results in an LDPC code.
Another example is expander codes \cite{sipser1996expander}. A binary expander code $\mathcal{C}$ is based on a bipartite expander graph $\Gamma: [N] \times [D]\rightarrow [M]$ with $N$ nodes on the left, $M$ nodes on the right, and left degree $D$. 
The parity check matrix of $\mathcal{C}$ is defined as follows. Each left node corresponds to a codeword bit and each right node corresponds to a parity check which checks if the parity of its neighboring codeword bits is $0$.
Such a code has linear time decoding algorithms, and the distance of $\mathcal{C}$ can be lower bounded by using the vertex expansion property of $\Gamma$. Specifically, assume that for every left set $A \subseteq [N]$, with $|A| \le \alpha N$, the neighbors of $A$, denoted as $\Gamma(A)$ has size at least $(1-\varepsilon)D|A|$, then \cite{10024861} showed that the distance of $\mathcal{C}$ is at least roughly $\frac{\alpha N}{2 \varepsilon}$.
Notice that an $\mathcal{S}$-shortening of $\mathcal{C}$ actually corresponds to deleting nodes in $\mathcal{S}$ from the left set $[N]$ together with their adjacent edges, thus this does not change the vertex expansion property of the remaining graph. Hence the new code still has a distance of at least roughly $\frac{\alpha N}{2 \varepsilon}$, which in fact corresponds to a larger relative distance (since the new code has a shorter length). As we will see shortly, this is actually a general property of any shortening of a code. In summary, just like puncturing, the shortening operation also preserves certain nice properties of the mother code, e.g., being an LDPC code or an expander code. In turn, this makes decoding easier.

Before stating our results, we first review some previous works on random puncturing and random shortening in more detail.

\subsection{Previous Work}
Recently, random  puncturing  has  drawn a lot of attention in the context of list decoding.
In \cite{wootters2013list}, Wootters showed that by applying a random puncturing to a Reed-Muller code and setting the desired rate to $O(\varepsilon^2)$,  with high probability one can list-decode the punctured code up to a relative radius of $1/2-\varepsilon$, with an exponential but non-trivial list size.
In \cite{rudra2014every}, Rudra and Wootters showed that if the mother code is an RS code, and has a large enough relative distance of $1-1/q - \varepsilon^2$, then after puncturing one can get a list-decoding radius of $ 1-1/q-\varepsilon$ and a rate close to capacity up to a $\poly \log (1/\varepsilon)$ factor, while the list size is $ O(1/\varepsilon)$.
We remark that a rate upper bound for list-decodable linear codes is given by Shangguan and Tamo \cite{shangguan2020combinatorial}, which is  a generalized singleton bound. 
Specifically, they proved that if $\mathcal{C}$ is a linear code of rate $R$ that is $(\rho, L)$ list decodable, i.e., the code has a relative list decoding radius of $\rho$ and list size $L$, then
$\rho\le  (1-R) \frac{L}{L+1}$.
They conjectured the existence of such codes and proved the case for $L=2,3$.
Later, towards proving this conjecture, Guo et.\ al.\ \cite{guo2022improved} showed that there are RS codes that are $(1-\varepsilon, O(1/\varepsilon))$ list decodable and the rate can be $\Omega(\varepsilon/\log(1/\varepsilon))$, though they mainly use intersection matrices instead of random puncturing. 
Ferber, Kwan, and Sauermann \cite{ferber2022list} further showed that 
through random puncturing one can achieve a rate of $ \varepsilon/15$ with list decoding radius $1-\varepsilon$ and list size $O(1/\varepsilon)$. 
This was further improved by Goldberg et.\ al.\ \cite{goldberg2021list}  to achieve a rate of $\frac{\varepsilon}{2-\varepsilon}$.
Most recently, \cite{brakensiek2022generic} showed that random puncturing of RS codes can go all the way up to the generalized singleton bound if the field size is $2^{O(n)}$, resolving a main conjecture of \cite{shangguan2020combinatorial}. This was subsequently improved by \cite{guo2023randomly}, which reduced the field size to $O(n^2)$; and again by \cite{alrabiah2023randomly}, which further reduced the field size to $O(n)$, although  \cite{guo2023randomly,alrabiah2023randomly} can only get close to the generalized singleton bound. We note that all the above works mainly studied RS codes or RM codes, which have strong algebraic structures, and some of them also require a large relative distance (e.g., close to $1-1/q$). 

On the other hand, Guruswami and Mosheiff \cite{guruswami2022punctured}  considered random puncturing of more general codes with weaker properties. 
Specifically, they considered two cases, where the mother code either has a low bias or has a large distance over a large alphabet (note that the property of a low bias implies a large distance, hence is stronger).
For both cases, they showed that the punctured code can achieve list decoding close to capacity. In fact, they showed a stronger result, that all monotone-decreasing local properties of the punctured code are similar to those of random linear codes.
Subsequent to  \cite{guruswami2022punctured}, Putterman and Pyne \cite{pyne2023pseudorandom} showed that the same results in \cite{guruswami2022punctured} can be achieved by using a pseudorandom puncturing instead, which reduces the number of random bits used in the puncturing to be linear in the block length of the punctured code, even if the mother code has a much larger length.

Unlike puncturing, there are only a handful of previous works on shortening.
In \cite{nelson2015existence}, Nelson and Van Zwam proved that all linear codes can be obtained by a sequence of puncturing and/or shortening of a collection of asymptotically good codes.
In \cite{yardi2017shortened}, Yardi and Pellikaan showed that all linear codes can be obtained by a sequence of puncturing and/or shortening on some specific cyclic code.
In \cite{bioglio2017low},  Bioglio et. al. presented a low-complexity construction of polar codes with arbitrary length and rate using shortening and puncturing.
In \cite{liu2021shortened}, Liu et. al. provided some general properties of shortened linear codes.

\subsection{Notation and Definitions.}
\begin{definition}\label{def_code}
    A \textit{linear code} $\mathcal{C}$ of length $n$ and dimension $k$ over a finite field $\mathbb{F}_q$ is a $k$-dimensional subspace of the $n$-dimensional vector space $\mathbb{F}_q^n$. The \textit{rate} of $\mathcal{C}$ is the ratio of the dimension to the length of the code: $R(\mathcal{C}) = \frac{k}{n}$. The \textit{distance} (or \textit{minimum distance}) of $\mathcal{C}$ is the minimum Hamming distance between any two distinct codewords in the code: $d(\mathcal{C}) = \underset{c_1, c_2 \in \mathcal{C}, c_1 \neq c_2}{\min} d(c_1, c_2)$. The \textit{relative distance} of $\mathcal{C}$ is the ratio of its distance to its length: $\delta(\mathcal{C})=\frac{d(\mathcal{C})}{n}$. 
    
    The dual code $\mathcal{C}^\perp$ of a linear code is the dual linear subspace of $C$. Hence the sum of the rates of $C$ and $C^{\perp}$ is 1. We call $d^\perp(\mathcal{C})$ the \textit{dual distance} of $\mathcal{C}$ as the minimum distance of its dual code $\mathcal{C}^\perp$. The \textit{relative dual distance} of $\mathcal{C}$ is the ratio of its dual distance to its length: $\delta^\perp(\mathcal{C}) = \frac{d^\perp(\mathcal{C})}{n}$. 
    
    We denote a linear code with these properties as an $[n,k,d]_q$ code or an $[n,k,d,d^{\perp}]_q$ code. Moreover, a linear code can be described by a $k \times n$ \textit{generator matrix} $G$, where each codeword in $\mathcal{C}$ is a linear combination of the rows of $G$. The \textit{parity check matrix} of $\mathcal{C}$ is an $n\times (n-k)$ matrix $H$ satisfying the property that for any codeword $c \in \mathcal{C}$, $cH = 0$. So $H^{\top}$ is the generator matrix of $\mathcal{C}^{\perp}$.
\end{definition}

\begin{definition}\label{def_puncturing}
    Let $\mathcal{P}$ be a subset of $[n]$ of size $p$. A $\mathcal{P}$-\textit{puncturing} on a code $\mathcal{C}$ of length $n$ involves removing $p$ positions indexed by $\mathcal{P}$. The resulted \textit{punctured code} $\mathcal{C}^{(\mathcal{P})}$ has length $n-p$. If $\mathcal{P}$ is a uniformly random subset of size $p$, we say that $\mathcal{C}^{(\mathcal{P})}$ is obtained from $\mathcal{C}$ by a random $p$-puncturing.
\end{definition}
    
\begin{definition}\label{def_shortening}
    Let $\mathcal{S}$ be a subset of $[n]$ of size $s$. An $\mathcal{S}$-\textit{shortening} on a code $\mathcal{C}$ of length $n$ involves selecting all codewords with zero values on positions indexed by $\mathcal{S}$ and removing these positions. The resulted \textit{shortened code} $\mathcal{C}^{[\mathcal{S}]}$ has length $n-s$. If $\mathcal{S}$ is a uniformly random subset of size $s$, we say that $\mathcal{C}^{[\mathcal{S}]}$ is obtained from $\mathcal{C}$ by a random $s$-shortening.
\end{definition}

\begin{definition}
The $q$-ary entropy function is defined as $\mathrm{H}_q(x) = x\log_q(q-1)-x\log_q x - (1-x)\log_q(1-x)$. 
\end{definition}

Throughout the paper, we use ``with high probability" to mean that when the rate $R$, relative distance $\delta$, relative dual distance $\delta^{\perp}$ of the code, and other given parameters are fixed, the probability of the event is $1 - O(\exp(-tn))$ for some constant $t$. Essentially, this means that the probability of the event occurring approaches 1 as the block length $n$ increases, making it increasingly likely that the desired properties hold.

As in \cite{guruswami2022punctured}, in this paper we also consider \emph{monotone-decreasing, local} properties. Informally, we call a code property $\mathscr{P}$ monotone-decreasing and local if, the fact that a code $\mathcal{C}$ does not satisfy $\mathscr{P}$ can be witnessed by a small ``bad set" of codewords in $\mathcal{C}$. For example, some typical properties, such as being list-decodable to capacity and having a small bias, are monotone-decreasing and local properties. More formally, a monotone-decreasing and local property is the opposite of a monotone-increasing and local property, defined below. 

\begin{definition}\label{def:monotonelocal}
A property $\mathscr{P}$ is said to be
\begin{itemize}
    \item \textit{monotone-increasing} if, for any code $\mathcal{C}$, whenever one of its subcodes (i.e., a subspace of $\mathcal{C}$) satisfies $\mathscr{P}$, the code $\mathcal{C}$ itself also satisfies $\mathscr{P}$ (\textit{monotone-decreasing} if the complement of $\mathscr{P}$ is monotone-increasing); 
    \item \textit{$b$-local} for some $b\in\mathbb{N}$ if there exists a family $\mathcal{B}_{\mathscr{P}}$ of sets of words in $\mathbb{F}_q^n$, with the size of the sets at most $b$, and such that $\mathcal{C}$ satisfies $\mathscr{P}$ if and only if 
    there exists an set $B\in\mathcal{B}_{\mathscr{P}}$ satisfying $B \subseteq \mathcal{C}$,
    \item \textit{row-symmetric} if, for any code $\mathcal{C} \subseteq \mathbb{F}_q^n$ that satisfies $\mathscr{P}$, the resulting code obtained by performing a permutation on the $n$ positions of $\mathcal{C}$ also satisfies $\mathscr{P}$.
\end{itemize}
\end{definition}

\subsection{Main Results}
\paragraph{Random puncturing vs. random shortening} Before formally stating our results, we first informally compare the two operations of random puncturing and random shortening. A random $p$-puncturing of a code of length $n$ involves uniformly selecting $p$ positions randomly from $[n]$, and discarding these positions in the code. One can see that under appropriate conditions, this operation preserves the distinctness of all codewords, and thus can increase the rate of the code. However it may decrease the distance or relative distance of the code. In contrast, a random $s$-shortening of a code involves picking $s$ positions uniformly randomly from $[n]$, forming a subcode that consists of codewords which contain only zeros at these positions, and then deleting these positions in the subcode. It can be seen that this operation perserves the distance of the code, and thus increases the relative distance of the code, but on the other hand the rate of the code can potentially decrease. Hence, these two operations are indeed ``dual" in some sense, and therefore one can apply both operations to adjust both the rate and the relative distance of the code.


A linear code $\mathcal{C}\subseteq \mathbb{F}_q^n$, where $q=p^r$ for some prime $p$, is called $\varepsilon$-biased, if for every codewords $c\in\mathcal{C}$, $\left|\sum_{i=1}^n \omega^{\mathrm{tr}\left(a \cdot c_i\right)}\right| \leq \varepsilon n$ for all $a \in \mathbb{F}_q^{\ast}$. where $\omega=e^{\frac{2 \pi i}{p}}$ and $\mathrm{tr}: \mathbb{F}_q \rightarrow \mathbb{F}_p$ is the field trace map. In Section \ref{biased_codes}, we will provide a more detailed explanation of the $\epsilon$-biased code.

Our main results show that random shortening is an effective way to reduce the \emph{bias} of a code. Note that this is stronger than increasing the relative distance, since the former implies the latter (see Proposition~\ref{pro_biase2wt}). If the mother code satisfies certain conditions, then we show after random shortening the new code can achieve an arbitrarily small bias with high probability. We note that a random linear code has a small bias, and thus in this sense the code after random shortening behaves like random linear codes. Moreover, the condition that the mother code has a low bias is required in several previous works (e.g., \cite{guruswami2022punctured,pyne2023pseudorandom}), while these works essentially do not care about the rate of the mother code. Thus we can apply a random puncturing to the new code after a random shortening, to get another code where all monotone-decreasing local properties are similar to those of random linear codes. This further weakens the requirements of mother codes in previous works to some extent.

\paragraph{Low-Biased codes from codes with large distance.}

A low-biased code must have a large distance, as stated in Proposition \ref{pro_biase2wt}. However, the reverse may not hold. The following theorem shows that it is also possible to derive a low-biased code from a code with a large distance by random shortening.

\begin{theorem}\label{thm_main1}
    For any $0<\varepsilon<1$, any $[n, Rn, \delta n]_q$ code $\mathcal{C}$ with $\frac{q-1}{q}-\frac{q}{q-1}\left(\frac{\varepsilon}{2(q-1)}\right)^2<\delta<\frac{q-1}{q}$ and any constant $0<\gamma<R$, there exists a number $0 < s < R$ such that the following holds. If we perform a random $sn$-shortening $\mathcal{S}$ to $\mathcal{C}$, then with high probability, the shortened code $\mathcal{C}^{[\mathcal{S}]}$ is $\varepsilon$-biased and has rate at least $R - \gamma$.
\end{theorem}


We note that the theorem only requires a lower bound on the relative distance, but there are no restrictions on the rate of the original code, $R$. Hence, this requirement is generally easy to satisfy, for example, from simple constructions using code concatenation. Furthermore, we can select an appropriate shortening size to ensure that the rate of the shortened code is arbitrarily close to $R$.

The distance condition of $\mathcal{C}$ in Theorem \ref{thm_main1} can also be relaxed, resulting in the following theorem.


\begin{theorem}\label{thm_main12}
    Given any $0 < \varepsilon < 1$, if an $[n, Rn, \delta n]_q$ code $\mathcal{C}$ satisfies the condition that there exists some $0<\beta<1$, such that $\frac{\delta}{1 - (1-\beta)R} > \frac{q - 1}{q} - \frac{q}{q - 1}\left(\frac{\varepsilon}{2(q - 1)}\right)^2$, then there exists a number $0 < s < R$ such that the following holds. If we perform a random $sn$-shortening $\mathcal{S}$ to $\mathcal{C}$, then with high probability, the shortened code $\mathcal{C}^{[\mathcal{S}]}$ is $\varepsilon$-biased with rate at least $\beta R$.
\end{theorem}


Indeed, the asymptotic form of the Plotkin bound is given by 
\begin{equation}
    R \leq 1-(\frac{q}{q-1})\cdot \delta +o(1).
\end{equation}
Thus Theorem \ref{thm_main12} implies that as long as the rate-distance trade-off of the original code is close enough to the Plotkin bound, we can obtain a code with an arbitrarily small bias by random shortening. On the other hand, unlike in Theorem \ref{thm_main1}, the rate of the shortened code may not be arbitrarily close to $R$, but we can still get a new rate that is only a constant factor smaller.

\begin{remark}
Just as in \cite{guruswami2022punctured}, in both our theorems, the condition that $\mathcal{C}$ is a linear code is crucial for ensuring that the shortened code maintains a constant rate with high probability. When $\mathcal{C}$ is non-linear, there are certain counterexamples. For instance, let $\mathcal{C'}$ be an $[n, Rn, \delta n]_q$ code and let $\mathcal{C}$ be obtained from $\mathcal{C}'$ by adding $\tau n$ $1$'s at the end of each codeword in $\mathcal{C}$. By picking an arbitrarily small constant $\tau>0$, the rate and relative distance of $\mathcal{C}$ are almost the same as $\mathcal{C}'$. However, after a random $sn$-shortening, $\mathcal{C}^{[\mathcal{S}]}$ is the $0$ code with probability $1-2^{-\Omega(n)}$.
\end{remark}

\paragraph{Low-biased codes from codes with small rate and not too small dual distance.}

In the next theorem, there is no requirement for $\delta$ to be very large. Instead, we impose constraints on its dual distance, $\delta^{\perp}$, and the rate, $R$. If the dual distance is not too small and the rate can be upper bounded, then we can also apply the shortening technique to obtain a low-biased code.


\begin{theorem}\label{thm_main2}
Given any $0<\varepsilon<1$, if an $[n,Rn,\delta n,\delta^{\perp}n]_q$ code $\mathcal{C}$ satisfies the condition that there exist $0<\gamma<\frac{1}{4}$, $0<\delta_0^{\perp}<\min\{\varepsilon^{\frac{1}{\gamma}},\left( \frac{1+\log_q(1-\delta)}{36}\right)^2,(\frac{1}{q})^{\frac{1}{\gamma}}\}$, such that $\delta^{\perp}>\delta_0^{\perp}$ and $0<R<\frac{0.5-2\gamma}{1+0.9\cdot\log_q(1-\delta)}\mathrm{H}_q(\delta_0^{\perp})$, then there exists a number $0<s<R$ such that the following holds. If we perform a random $sn$-shortening $\mathcal{S}$ to $\mathcal{C}$, then with high probability the shortened code $\mathcal{C}^{[\mathcal{S}]}$ is $\varepsilon$-biased with rate at least $0.1R$.
\end{theorem}

In Theorem \ref{thm_main2}, the rate of the dual code must be sufficiently large.\ Additionally, if the term $\frac{\frac{1}{2}-2\gamma}{1+0.9\cdot\log_q(1-\delta)}$ is less than 1, the rate-distance trade-off of the dual code surpasses the Gilbert-Varshamov (GV) bound. Consequently, when examining the problem within the context of the GV bound, we need to impose specific constraints on $\delta$. This leads to the following corollary.

\begin{corollary}\label{cor_main2cor}
    Given any $0<\varepsilon<1$, $\delta>1-q^{-0.6}$, there exists a number $\gamma>0$, such that for any $\delta^{\perp}>\delta_0^{\perp}$, $0<R<(1+\gamma)\mathrm{H}_q(\delta_0^{\perp})$ for a certain $0<\delta_0^{\perp}<\min\{\varepsilon^{\frac{1}{\gamma}},\frac{1}{8100},(\frac{1}{q})^{\frac{1}{\gamma}}\}$, there exists a number $0<s<R$ such that the following holds. Let $\mathcal{C}$ be any $[n, Rn, \delta n, \delta^{\perp}n]_q$ code. If we perform a random $sn$-shortening $\mathcal{S}$ to $\mathcal{C}$, then with high probability the shortened code $\mathcal{C}^{[\mathcal{S}]}$ is $\varepsilon$-biased  with rate at least $0.1R$.
\end{corollary}

Theorem \ref{thm_main2} and Corollary \ref{cor_main2cor} show that as long as the mother code and its dual both have a reasonable relative distance, one can use random shortening to get a new code with an arbitrary small bias, while only losing a constant factor in the rate.\ We note that linear codes such that both the code and its dual have good relative distance are also easily constructible, for example, see \cite{10.1007/s00037-009-0281-5}.

\paragraph{Random-like codes by random shortening and puncturing.}

In \cite{guruswami2022punctured}, the authors showed that a random puncturing of a low-biased code results in a new code that behaves like random linear codes. Using our theorems, we present a weaker condition that still achieves similar results. This follows from a combination of random shortening and random puncturing, as briefly discussed before.


\begin{theorem}\label{thm_main3}
For any $0 < \varepsilon < 1$, $b \in \mathbb{N}$, and prime power $q$, there exists some $\eta > 0$, such that the following holds. Let $\mathscr{P}$ be a monotone-decreasing, $b$-local, and row-symmetric property over $\mathbb{F}_q^n$ satisfied by a random linear code of length $n$ and rate $R^{\prime}$. There exists some $\eta > 0$ such that the following holds. If any one of the following properties is satisfied for  $R,\delta,\delta^{\perp},q,\eta$:
\begin{enumerate}
\item $\delta > (\frac{q - 1}{q} - \eta)(1 - R)$, or
\item $\delta^{\perp}>\delta_0^{\perp}$ and $0<R<\frac{\frac{1}{2}-2\gamma}{1+0.9\cdot\log_q(1-\delta)}\mathrm{H}_q(\delta_0^{\perp})$ for a certain $0<\delta_0^{\perp}<\min\{\varepsilon^{\frac{1}{\gamma}},\left( \frac{1+\log_q(1-\delta)}{36}\right)^2,(\frac{1}{q})^{\frac{1}{\gamma}}\}$, 
\end{enumerate}
then there exists $m,p,s>0$ such that for any $[m, Rm, \delta m]_q$ code, if we perform a random $sm$-shortening and then a random $pm$-puncturing on $\mathcal{C}$, the resulted code $\mathcal{D}$ has length $n$, rate at least $R^{\prime}-\varepsilon$ and with high probability, satisfies $\mathscr{P}$.
\end{theorem}

\begin{remark}
In fact, all our theorems hold under the more restricted \emph{pseudorandom} shortening, where $\mathcal{S}$ is sampled from a random walk on a sufficiently good constant-degree expander graph, as in \cite{pyne2023pseudorandom}. The reason is that in this case, the only thing that changes is the probability that a random shortening does not hit a codeword. A standard hitting set property of random walks on expander graphs ensures that this probability (as in Lemma \ref{lem_hitprob}) remains close to $(1-\delta)^s$, as long as the second largest normalized eigenvalue in absolute value of the expander is small enough, which can be achieved by having a large enough constant degree. Consequently, all our proofs essentially go through with only minimal changes. Thus, like \cite{pyne2023pseudorandom}, we can also reduce the number of random bits used in the shortening to be linear in the block length of the shortened code, even if the mother code has a much larger length.
\end{remark}

\begin{remark}
We emphasize that, in all our theorems, the conditions of the mother code we need are \emph{weaker} than the conditions in previous works which only use random puncturing \cite{guruswami2022punctured,pyne2023pseudorandom}. For example, when the alphabet size is small, those works require the mother code to have a small bias, which implies a large distance. On the other hand, our theorems either only require a large distance, or only require a reasonable distance in both the mother code and its dual. 
\end{remark}


\paragraph{Discussions and open questions.} Our work leaves several open questions for future investigation. One such question is whether we can achieve a good rate when performing random shortening. Right now, using our analysis, the rate of the code after random shortening is potentially worse than the mother code we start with, and thus we do not get a good rate-distance trade-off by simply applying random shortening. Therefore, it is a natural question to see if one can  achieve an $\varepsilon$-biased code through random shortening while at the same time maintaining a favorable $R$--$\varepsilon$ trade-off, if we start with some good initial conditions of the mother code. Another direction is further weakening the initial conditions required for obtaining low-biased codes, such as removing the constraint of $\delta$ in Theorem \ref{thm_main2}. In our view, these questions present exciting opportunities for advancing our understanding of random shortening, low-biased codes and their connections. 

Finally, it would also be quite interesting to completely derandomize the random shortening, thus yielding explicit constructions of low-biased codes with special structures (e.g., LDPC codes or expander codes). This will complement existing constructions of low-biased codes, and possibly lead to more efficient decoding algorithms.

\
\subsection{Technique Overview}

We investigate the effect of shortening as follows. An $\mathcal{S}$-shortening applied to a code $\mathcal{C}$ of length $n$ involves selecting all codewords with zeros at positions indexed by $\mathcal{S}$ and removing these positions. Specifically, if the support of a codeword $c \in \mathcal{C}$ intersects $\mathcal{S}$ (in which case we say $\mathcal{S}$ hits $c$), then $c$ will not be included in the shortened code; if the support of $c$ does not intersect $\mathcal{S}$, then  there is a codeword $c^{\prime} \in \mathcal{C}^{[\mathcal{S}]}$, which is obtained from $c$ by removing all positions in $\mathcal{S}$. In this way, under a random shortening, each non-zero codeword has a certain probability of being dropped and a certain probability of being retained in $\mathcal{C}^{[\mathcal{S}]}$. If the distance of $\mathcal{C}$ is $\delta n$, then the probability of each codeword being hit and dropped is at least $1 - (1 - \delta)^{s}$ by Lemma \ref{lem_hitprob}, where $s$ is the size of $\mathcal{S}$.

We use $\mathcal{C}_{\varepsilon}$ to denote all codewords in $\mathcal{C}$ which are not $\varepsilon$-biased. If the size of $\mathcal{C}_{\varepsilon}$ is small, then by a union bound, the probability that not all codewords in $\mathcal{C}_\varepsilon$ are hit by $\mathcal{S}$ is exponentially small. Thus, with high probability, all codewords in $\mathcal{C}$ that are not hit by $\mathcal{S}$ and inherited to $\mathcal{C}^{[\mathcal{S}]}$ are $\varepsilon$-biased. Hence, a critical part of all our proofs is to upper bound the size of $\mathcal{C}_{\varepsilon}$.

Furthermore, as long as $\mathcal{C}$ is a linear code and $s$ is less than the dimension $k$ of $\mathcal{C}$, we know that the shortened code $\mathcal{C}^{[\mathcal{S}]}$ has dimension at least $k-s$. Consequently, $\mathcal{C}^{[\mathcal{S}]}$ retains a nonzero constant rate as well.

\paragraph{Change of parameters.} 
The shortening results in changes to the parameters of the code. Here, we mainly apply shortening for two purposes: adjusting the bias and amplifying the relative distance.

\begin{enumerate}
    \item \textbf{Adjusting the bias:} Let $\mathcal{C}$ be of length $n$.\ When $\mathcal{C}_{\varepsilon^{\prime}}$ is hit by $\mathcal{S}$, it implies that the codewords in $\mathcal{C}$ not hit by $\mathcal{S}$ are all $\varepsilon^{\prime}$-biased. However, it doesn't directly imply that the shortened code $\mathcal{C}^{[\mathcal{S}]}$ is also $\varepsilon^{\prime}$-biased, since the shortening operation changes the length of the code. Nevertheless, the new bias $\varepsilon$ of $\mathcal{C}^{[\mathcal{S}]}$ is given by $\varepsilon \leq \frac{\varepsilon^{\prime} n + s}{n - s}$, where $s$ is the size of the shortening $\mathcal{S}$. If $s$ is small compared to $n$, $\varepsilon$ is close to $\varepsilon^{\prime}$. In the proof of Theorem \ref{thm_main1}, we can choose $s$ to be a sufficiently small fraction of $n$. In the proof of Theorem \ref{thm_main2}, we provide an upper bound for $R$, which also enables us to choose a small shortening size. In both cases, we set the shortening size to be less than $0.05\varepsilon'n$, allowing us to choose $\varepsilon^{\prime} = 0.9 \varepsilon$.
    \item \textbf{Amplifying the relative distance:} We use another technique in the proof of Theorem \ref{thm_main12} to first transform a code with a rate-distance trade-off near the Plotkin bound into a code with near-optimal distance. By Proposition \ref{pro_dist}, the distance of the shortened code $\mathcal{C}^{[\mathcal{S}]}$ is no less than that of the original code $\mathcal{C}$. However, since $\mathcal{C}^{[\mathcal{S}]}$ has length $n - s$ instead of $n$, its relative distance becomes $\frac{\delta}{1 - \frac{s}{n}}$. This allows us to increase the relative distance of the code. In turn, Theorem \ref{thm_main12} follows from Theorem \ref{thm_main1}. 
\end{enumerate}

\paragraph{Estimation of the size of $\mathcal{C}_\varepsilon$.}

This is the most critical part of all our proofs. For Theorem \ref{thm_main1} and Theorem \ref{thm_main2}, we have two different ways of estimating the upper bound of $|\mathcal{C}_\varepsilon|$:

\begin{enumerate}
    \item \textbf{Estimating $|\mathcal{C}_\varepsilon|$ with relative distance $\delta$:} We use $J_q(\delta)$ to denote the list decoding radius corresponding to the classical Johnson bound for a code over $\mathbb{F}_q$ with relative distance $\delta$. It is easy to see that when $\delta$ is close to the optimal $\frac{q-1}{q}$, so is $J_q(\delta)$. To give an upper bound of $|\mathcal{C}_\varepsilon|$, we construct $q$ balls in $\mathbb{F}_q^n$ with radius $J_q(\delta)$ and centered at $t\cdot\vec{1}$, where $\vec{1}$ is the all-one vector and $t\in \mathbb{F}_q$. By the Johnson bound, the number of codewords covered by these balls is at most $\mathrm{poly}(n)$. We show that, if a codeword $c$ is not covered by these balls, its empirical distribution over $\mathbb{F}_q$ is close to the uniform distribution, which implies $c$ is small biased. This upper bounds $|\mathcal{C}_\varepsilon|$ by $\mathrm{poly}(n)$.
    \item \textbf{Estimating $|\mathcal{C}_\varepsilon|$ with relative dual distance $\delta^{\perp}$ and rate $R$:} If $\mathcal{C}$ has dual distance $d^{\perp}$, then any $d^{\perp}-1$ columns of the generator matrix of $\mathcal{C}$ are linearly independent, which means that if we uniformly randomly choose a codeword from $\mathcal{C}$, then any $d^{\perp}-1$ symbols of the codeword are independently uniform, i.e., the symbols of a random codeword are $d^{\perp}-1$-wise independent. We can now use this property to estimate the probability that a codeword randomly chosen from $\mathcal{C}$ is not $\varepsilon$-biased. This is a typical application of the concentration phenomenon from the higher moment method, where we use Hoeffding inequality, Chernoff bound, and Sub-Gaussian property to bound the $(d^{\perp}-1)$th moment of the summation of some random variables. Then by Markov's inequality, the probability that a random codeword is not $\varepsilon$-biased can be bounded, which also gives an upper bound on $|\mathcal{C}_\varepsilon|$. 
\end{enumerate}

\paragraph{Obtaining random-like codes.}
To obtain random-like codes, we combine our results with those in \cite{guruswami2022punctured}, which state that a randomly punctured low-biased code is likely to possess any monotone-decreasing local property typically satisfied by a random linear code of a similar rate.\ By our results, we can start from a code with less stringent conditions and achieve the same results as in \cite{guruswami2022punctured}, through the operations of a random shortening followed by a random puncturing. 

\paragraph{Organization.}
The rest of this paper is organized as follows. In Section \ref{sec_pre}, we describe some basic definitions, terms, and useful properties. In Section \ref{sec_ran}, we show how to combine other theorems and the work of \cite{guruswami2022punctured} to obtain a random-like code with weaker initial conditions and prove Theorem \ref{thm_main3}. In Section \ref{sec_est}, we present two methods for estimating $|\mathcal{C}_\varepsilon|$. In Section \ref{sec_prove1}, we prove Theorem \ref{thm_main1} and Theorem \ref{thm_main12}. In Section \ref{sec_pro2}, we prove Theorem \ref{thm_main2} and Corollary \ref{cor_main2cor}.

%% file: prelim.tex
\section{Preliminary}\label{sec_pre}

\subsection{Punctured codes and shortened codes}\label{punc_short}
\begin{definition}\label{def_code}
    A \textit{linear code} $\mathcal{C}$ of length $n$ and dimension $k$ over a finite field $\mathbb{F}_q$ is a $k$-dimensional subspace of the $n$-dimensional vector space $\mathbb{F}_q^n$. The \textit{rate} of $\mathcal{C}$ is defined as the ratio of the dimension to the length of the code: $R(\mathcal{C}) = \frac{k}{n}$. The \textit{distance} (or \textit{minimum distance}) of $\mathcal{C}$ is the minimum Hamming distance between any two distinct codewords in the code: $d(\mathcal{C}) = \min_{c_1, c_2 \in \mathcal{C}, c_1 \neq c_2} d(c_1, c_2)$. The \textit{relative distance} of $\mathcal{C}$ is the ratio of its distance to its length: $\delta(\mathcal{C})=\frac{d(\mathcal{C})}{n}$. The \textit{dual distance} of $\mathcal{C}$ is the minimum distance of its dual code $\mathcal{C}^\perp$, denoted by $d^\perp(\mathcal{C})$. The \textit{relative dual distance} of $\mathcal{C}$ is the ratio of its dual distance to its length: $\delta^\perp(\mathcal{C}) = \frac{d^\perp(\mathcal{C})}{n}$. We denote a linear code with these properties as an $[n,k,d]_q$ code or an $[n,k,d,d^{\perp}]_q$ code. A linear code can be described by a $k \times n$ \textit{generator matrix} $G$, where each codeword in $\mathcal{C}$ can be obtained as a linear combination of the rows of $G$. The \textit{parity check matrix} of $\mathcal{C}$ is an $n\times (n-k)$ matrix $H$ satisfying the property that for any codeword $c \in \mathcal{C}$, $cH = 0$.
\end{definition}

\begin{definition}\label{def_puncturing}
    Let $\mathcal{P}$ be a subset of $[n]$ of size $p$. A $\mathcal{P}$-\textit{puncturing} on a code $\mathcal{C}$ of length $n$ involves removing $p$ positions indexed by $\mathcal{P}$. The resulted \textit{punctured code} $\mathcal{C}^{(\mathcal{P})}$ has length $n-p$. If $\mathcal{P}$ is a uniformly random subset of size $p$, we say that $\mathcal{C}^{(\mathcal{P})}$ is obtained from $\mathcal{C}$ by a random $p$-puncturing.
\end{definition}
    
\begin{definition}\label{def_shortening}
    Let $\mathcal{S}$ be a subset of $[n]$ of size $s$. An $\mathcal{S}$-\textit{shortening} on a code $\mathcal{C}$ of length $n$ involves selecting all codewords with zero values on positions indexed by $\mathcal{S}$ and removing these positions. The resulted \textit{shortened code} $\mathcal{C}^{[\mathcal{S}]}$ has length $n-s$. If $\mathcal{S}$ is a uniformly random subset of size $s$, we say that $\mathcal{C}^{[\mathcal{S}]}$ is obtained from $\mathcal{C}$ by a random $s$-shortening.
\end{definition}

\begin{proposition}\label{pro_matrix}
For a linear code $\mathcal{C}$, and $\mathcal{P},\mathcal{S}\subseteq [n]$, $\mathcal{C}^{(\mathcal{P})}$ and $\mathcal{C}^{[\mathcal{S}]}$ are also both linear codes. The generator matrix of $\mathcal{C}^{(\mathcal{P})}$ is obtained from that of $\mathcal{C}$ by deleting $p$ columns indexed by $\mathcal{P}$, and the parity check matrix of $\mathcal{C}^{[\mathcal{S}]}$ is obtained from that of $\mathcal{C}$ by deleting $s$ rows indexed by $\mathcal{S}$.
\end{proposition}

Given that the parity check matrix of a code is the transpose of the generator matrix of its corresponding dual code, we can deduce that shortening a code is equivalent to puncturing its dual code. Consequently, we can observe the following properties.

\begin{proposition}\label{pro_dual}
For a linear code $\mathcal{C}$, and $\mathcal{P}\subseteq [n]$, $(\mathcal{C}^{(\mathcal{P})})^{\perp}=(\mathcal{C}^{\perp})^{[\mathcal{P}]}$.
\end{proposition}

\begin{definition}\label{def_hit}
    Let $\mathcal{C}\subseteq\mathbb{F}_q^n$ and $c$ be a codeword in $\mathcal{C}$, and $\mathcal{S}\subseteq [n]$. Denote $\mathrm{supp}(c)$ to be the set of non-zero coordinates of $c$. We say $\mathcal{S}$ \textit{hits} $c$ if $\mathrm{supp}(c)\cap \mathcal{S} \neq \varnothing$, $\mathcal{S}$ hits $\mathcal{C}_{\varepsilon}$ if each codeword in $\mathcal{C}_{\varepsilon}$ is hit by $\mathcal{S}$.
\end{definition}

In general, when applying the shortening $\mathcal{S}$ to a code $\mathcal{C}$, we exclude any codeword that is hit by $\mathcal{S}$. Specifically, if a codeword $c\in\mathcal{C}$ is hit by $\mathcal{S}$, by definition we do not include it in the shortened code $\mathcal{C}^{[\mathcal{S}]}$. On the other hand, if a codeword $c$ is not hit by $\mathcal{S}$, there exists a corresponding codeword $c^{\prime}$ in the shortened code $\mathcal{C}^{[\mathcal{S}]}$, which is obtained by removing all positions indexed by $\mathcal{S}$.

Next, we examine the effects of the shortening operation on code parameters.

\begin{proposition}\label{pro_rate}
    Let $\mathcal{C}$ be an $[n,k,d, d^{\perp}]_q$ code and $\mathcal{S}$ be a subset of $[n]$ of size $s<k$. The shortened code $\mathcal{C}^{[\mathcal{S}]}$ has dimension $k^{\prime}\geq k-s$. Moreover, if $s<d^{\perp}$, then $k^{\prime}= k-s$.
\end{proposition}

\begin{proof}
    By Proposition \ref{pro_matrix}, the parity matrix of $\mathcal{C}^{[\mathcal{S}]}$ is $(n-s)\times (n-k)$, which implies that the dimension of $\mathcal{C}^{[\mathcal{S}]}$ is at least $(k-s)$. If $s<d^{\perp}$, by Proposition \ref{pro_dual}, there doesn't exist a non-zero codeword $c\in\mathcal{C}^{\perp}$ such that $\mathrm{supp}(c) \subseteq \mathcal{S}$. In this case, there won't be any collisions between two codewords in $\mathcal{C}^{\perp}$ after shortening, so $\dim(\mathcal{C})^{\perp}=\dim(\mathcal{C}^{[\mathcal{S}]})^{\perp}=n-k$, and the punctured code $\mathcal{C}^{[\mathcal{S}]}$ has rate $k-s$.
\end{proof}

\begin{proposition}\label{pro_dist}
    Let $\mathcal{C}$ be an $[n,k,d, d^{\perp}]_q$ code and $\mathcal{S}$ be a subset of $[n]$ of size $s<k$, The distance of shortened code $\mathcal{C}^{[\mathcal{S}]}$ is at least $d$.
\end{proposition}

\begin{proof}
    By definition, any non-zero codeword $c\in\mathcal{C}^{[\mathcal{S}]}$ comes from some codewords $c^{\prime}\in\mathcal{C}$ by removing all coordinates indexed by $[\mathcal{S}]$, and $c_i=0$ for any $i\in \mathcal{S}$, and thus $\mathrm{wt}(c)\geq \mathrm{wt}(c^{\prime})\geq d$. So, the distance of $\mathcal{C}^{[\mathcal{S}]}$ is at least $d$.
\end{proof}

Moreover, since the shortened code $\mathcal{C}^{\mathcal{[S]}}$ has length of $n-s$ instead of $n$, its relative distance becomes $\frac{\delta}{1-\frac{s}{n}}$. This allows us to increase the relative distance of the code.

\subsection{A characterization of $\varepsilon$-biased code}\label{biased_codes}
In this section, we recall the definition of $\varepsilon$-biased code.
\begin{definition}\label{def_biased}
    Let $\mathcal{C}\subseteq \mathbb{F}_q^n$, where $q=p^r$ for some prime $p$ and let $\varepsilon>0$. A vector $x \in \mathbb{F}_q^n$ is said to be $\varepsilon$-\textit{biased} if $\left|\sum_{i=1}^n \omega^{\mathrm{tr}\left(a \cdot x_i\right)}\right| \leq \varepsilon n$ for all $a \in \mathbb{F}_q^{\ast}$. Here, $\omega=e^{\frac{2 \pi i}{p}}$ and $\mathrm{tr}: \mathbb{F}_q \rightarrow \mathbb{F}_p$ is the field trace map: $\mathrm{tr}(x)=\sum_{i=1}^{r}x^{p^{i}}$. The code $\mathcal{C}$ is said to be $\varepsilon$-biased if every $c \in \mathcal{C} \backslash\{0\}$ is $\varepsilon$-biased.
\end{definition}

Note that for a binary vector $x\in\mathbb{F}_2^n$, $\left|\sum_{i=1}^n \omega^{\mathrm{tr}\left(- x_i\right)}\right|=\left|n-2\cdot\mathrm{wt}(x)\right|$. Then, a binary code $\mathcal{C}$ is $\varepsilon$-biased if and only of all non-zero codewords $c$ of $\mathcal{C}$ has weight $\frac{1-\varepsilon}{2}n\leq \mathrm{wt}(c)\leq \frac{1+\varepsilon}{2}n$.

\begin{proposition}\label{pro_biase2wt}
    If $\mathcal{C}\subseteq \mathbb{F}_q^n$ is $\varepsilon$-biased, then its distance is at least $\frac{(q-1)(1-\varepsilon)}{q}n$.
\end{proposition}

\begin{proof}
    Since $\mathcal{C}\subseteq \mathbb{F}_q^n$ is $\varepsilon$-biased, for any $c\in\mathcal{C}$,
    \begin{equation}
        \left|\sum_{a\in \mathbb{F}_q^{\ast}}\sum_{i=1}^{n}\omega^{\mathrm{tr}\left(a \cdot c_i\right)} \right| = \left|\sum_{a\in \mathbb{F}_q^{\ast}}\left(\sum_{i: c_i\neq 0}\omega^{\mathrm{tr}\left(a \cdot c_i\right)} + (n-\mathrm{wt(c)})\right)\right| \leq (q-1)\varepsilon n.
    \end{equation}
    By changing the order of double summation,
    \begin{equation}
        \left|\sum_{i: c_i\neq 0}\sum_{a\in \mathbb{F}_q^{\ast}}\omega^{\mathrm{tr}\left(a \cdot c_i\right)} + (q-1)(n-\mathrm{wt(c)})\right|\leq(q-1)\varepsilon n.
    \end{equation}
    Since $\sum_{a\in \mathbb{F}_q^{\ast}}\omega^{\mathrm{tr}\left(a \cdot c_i\right)}=-1$ for all $c_i\neq0$, we get
    \begin{equation}
        \left|(q-1)n-q\cdot\mathrm{wt}(c)\right|\leq (q-1)\varepsilon n,
    \end{equation}
    which implies $\mathrm{wt}(c)\geq \frac{(q-1)(1-\varepsilon)}{q}n$.
\end{proof}

\begin{definition}
    Given a vector $x\in\mathbb{F}_q^n$, we define its empirical distribution $\mathrm{Emp}_x$ over $\mathbb{F}_q$ by 
    \begin{equation}
        \mathrm{Emp}_x(t) = \mathrm{Pr}_{i\in[n]}(x_i=t).
    \end{equation}
\end{definition}

\begin{lemma}\label{lem_uniform}
    Given a vector $x\in\mathbb{F}_q^n$, if for any $t\in \mathbb{F}_q$,
    \begin{equation}
    \mathrm{Emp}_x(t) \leq \frac{1}{q} + \frac{\varepsilon}{2(q-1)}, 
    \end{equation}
    then $x$ is $\varepsilon$-biased.
\end{lemma}

\begin{proof}
We compute the total variation distance between $\mathrm{Emp}_x$ and the uniform distribution over $\mathbb{F}_q$, which is given by
    \begin{equation}
    \frac{1}{2} \sum_{t\in \mathbb{F}_q} |\mathrm{Emp}_x(t)-\frac{1}{q}|=\sum_{t:\mathrm{Emp}_x(t)>\frac{1}{q}}(\mathrm{Emp}_x(t)-\frac{1}{q})\leq  \frac{\varepsilon}{2}.
    \end{equation}
    Since for any $a \in \mathbb{F}_q^{\ast}$, 
    \begin{equation}
    \sum_{t\in \mathbb{F}_q} \omega^{\mathrm{tr}\left(a \cdot t\right)}=0.
    \end{equation}
    We know that
    \begin{equation}
    \begin{aligned}
    &\left|\sum_{i=1}^n \omega^{\mathrm{tr}\left(a \cdot x_i\right)}\right| \\ 
    = &\left|  \sum_{t\in\mathbb{F}_q}\sum_{i: x_i=t}\omega^{\mathrm{tr}\left(a \cdot x_i\right)}       \right| \\ 
    = &\left|  \sum_{t\in\mathbb{F}_q}\left(\sum_{i: c_i=t}\omega^{\mathrm{tr}\left(a \cdot x_i\right)} -\frac{n}{q} \omega^{\mathrm{tr}\left(a \cdot t\right)} \right)     \right| \\ 
    \leq &\sum_{t\in \mathbb{F}_q} n\cdot \left| \mathrm{Emp}_x(t) -\frac{1}{q}\right|\leq \varepsilon n,
    \end{aligned}
    \end{equation}
    which implies $x$ is $\varepsilon$-biased.
\end{proof}

\begin{definition}\label{def_epsiset}
    Let $\mathcal{C}$ be a code of length $n$. Denote $\mathcal{C}_{\varepsilon}$ to be the set of all non-zero codewords which are not $\varepsilon$-biased.
\end{definition}

\begin{lemma}\label{lem_eps}
    Let $\mathcal{C}$ be a code of length $n$, and $\mathcal{S}$ be a subset of $[n]$ of size $s$. If $\mathcal{S}$ hits $\mathcal{C}_{\varepsilon}$, then the shortened code $\mathcal{C}^{[\mathcal{S}]}$ is $\frac{\varepsilon n+s}{n-s}$-biased.
\end{lemma}

\begin{proof}
    Any non-zero codeword $c\in\mathcal{C}^{[\mathcal{S}]}$ comes from some codewords $c^{\prime}\in\mathcal{C}$ by removing all coordinates indexed by $[\mathcal{S}]$, and $c_i=0$ for any $i\in \mathcal{S}$. Since $\mathcal{S}$ hits $\mathcal{C}_{\varepsilon}$, $c^{\prime}$ is $\varepsilon$-biased, that is, for any $a\in \mathbb{F}_q^{\ast}$, $\left|\sum_{i=1}^n \omega^{\mathrm{tr}\left(a \cdot c^{\prime}_i\right)}\right| \leq \varepsilon n$. Then
    \begin{equation}
    \begin{aligned}
    &\left|\sum_{i=1}^{n-s} \omega^{\mathrm{tr}\left(a \cdot c_i\right)} \right| \\
    =&\left|\sum_{i\notin \mathcal{S}} \omega^{\mathrm{tr}\left(a \cdot c^{\prime}_i\right)} \right|\\
    \leq &\left| \sum_{i=1}^{n} \omega^{\mathrm{tr}\left(a \cdot c^{\prime}_i\right)}\right| + \left|\sum_{i \in \mathcal{S}} \omega^{\mathrm{tr}\left(a \cdot c^{\prime}_i\right)}  \right|\\
    \leq & \varepsilon n + s =\left(\frac{\varepsilon n+s}{n-s}\right) \left(n-s\right)
    \end{aligned}
    \end{equation}
    for any $a\in \mathbb{F}_q^{\ast}$, and thus $\mathcal{C}^{[\mathcal{S}]}$ is $\frac{\varepsilon n+s}{n-s}$-biased.
\end{proof}

\begin{lemma}\label{lem_hitprob}
    Let $\mathcal{C}$ be a code of length $n$ and distance $d=\delta n$, and $\mathcal{S}$ be a subset of $[n]$ of size $s$. For any codeword $c\in\mathcal{C}$, the probability that $c$ is not hit by $\mathcal{S}$ is at most
    $$
    (1-\delta)^{s}.
    $$
\end{lemma}

\begin{proof}
    Suppose $\mathrm{wt}(c)=i$. $c$ not hit by $\mathcal{S}$ means $\mathcal{S}\cap \mathrm{supp}(c)=\varnothing$, whose probability is at most
\begin{equation}
\frac{{n-i \choose s}}{{n \choose s}}\leq (1-\frac{i}{n})^{s}\leq (1-\delta)^{s} .
\end{equation}
\end{proof}

%% file: result.tex
\section{Random-like codes by random shortening and puncturing}\label{sec_ran}

In this section, we integrate the applications of low-biased codes with our theorems to obtain further results. In \cite{guruswami2022punctured}, random punctured low biased codes are studied, and they are shown to be locally similar to random linear codes. By combining this result, we can derive weaker conditions for obtaining a random-like code. We first recall some of the framework for studying properties of codes, which was established in \cite{mosheiff2020ldpc,guruswami2020sharp}.

A property $\mathscr{P}$ of length-$n$ linear codes over $\mathbb{F}_q$ is a collection of linear codes in $\mathbb{F}_q^n$. A linear code $\mathcal{C} \subseteq \mathbb{F}_q^n$ such that $\mathcal{C} \in \mathscr{P}$ is said to satisfy $\mathscr{P}$. 

\begin{definition}\label{def:monotonelocal}
A property $\mathscr{P}$ is said to be
\begin{itemize}
    \item \textit{monotone-increasing} if, for any code $\mathcal{C}$, whenever one of its subcodes (i.e., a subspace of $\mathcal{C}$) satisfies $\mathscr{P}$, the code $\mathcal{C}$ itself also satisfies $\mathscr{P}$ (\textit{monotone-decreasing} if the complement of $\mathscr{P}$ is monotone-increasing); 
    \item \textit{$b$-local} for some $b\in\mathbb{N}$ if there exists a family $\mathcal{B}_{\mathscr{P}}$ of sets of words in $\mathbb{F}_q^n$, with the size of the sets at most $b$, and such that $\mathcal{C}$ satisfies $\mathscr{P}$ if and only if 
    there exists an set $B\in\mathcal{B}_{\mathscr{P}}$ satisfying $B \subseteq \mathcal{C}$,
    \item \textit{row-symmetric} if, for any code $\mathcal{C} \subseteq \mathbb{F}_q^n$ that satisfies $\mathscr{P}$, the resulting code obtained by performing a permutation on the $n$ positions of $\mathcal{C}$ also satisfies $\mathscr{P}$.
\end{itemize}
\end{definition}

\begin{definition}
Let $\mathscr{P}$ be a nonempty monotone-increasing property over $\mathbb{F}_q^n$. Its threshold is defined by
\begin{equation}
\mathrm{RLC}(\mathscr{P})=\min \left\{R \in[0,1] \mid \mathrm{Pr}\left[\mathcal{C}_{\mathrm{RLC}}^{n, q}(R) \text { satisfies } \mathscr{P}\right] \geq \frac{1}{2}\right\},
\end{equation}
where $C_{\mathrm{RLC}}^{n, q}(R)$ is a random linear code of rate $R$ in $\mathbb{F}_q^n$.
\end{definition}

\begin{theorem}\cite{mosheiff2020ldpc}
Let $\mathscr{P} \subseteq \mathbb{F}_q^n$ be a random linear code of radius $R$ and Let $\mathscr{P}$ be a monotone-increasing, b-local and row-symmetric property over $\mathbb{F}_q^n$, where $\frac{n}{\log _q n} \geq \omega_{n \rightarrow \infty}\left(q^{2 b}\right)$. The following now holds for every $\varepsilon>0$.
\begin{enumerate}
    \item If $R \leq \mathrm{RLC}(\mathscr{P})-\varepsilon$ then
\begin{equation}
\mathrm{Pr}[\mathcal{C} \text { satisfies } \mathscr{P}] \leq q^{-\left(\varepsilon-o_n \rightarrow \infty(1)\right) n}.
\end{equation}
    \item If $R \geq \mathrm{RLC}(\mathscr{P})+\varepsilon$ then
\begin{equation}
\mathrm{Pr}[\mathcal{C} \text { satisfies } \mathscr{P}] \geq 1-q^{-\left(\varepsilon-o_{n \rightarrow \infty}(1)\right) n}.
\end{equation}
\end{enumerate}
\end{theorem}

\begin{theorem}\cite{guruswami2022punctured}\label{thm:GM22}
Let $q$ be a prime power, and let $\mathscr{P}$ be a monotone-decreasing, row-symmetric and b-local property over $\mathbb{F}_q^n$, where $\frac{n}{\log n} \geq \omega_{n \rightarrow \infty}\left(q^{2 b}\right)$. Let $\mathcal{D} \subseteq \mathbb{F}_q^m$ be a linear code. Let $\mathcal{C}$ be a random n-puncturing of $\mathcal{D}$ of design rate $R \leq \mathrm{RLC}(\mathscr{P})-\varepsilon$ for some $\varepsilon>0$. Suppose that $\mathcal{D}$ is $\left(\frac{\varepsilon b \ln q}{q^b}\right)$-biased. Then,
$$
\mathrm{Pr}[\mathcal{C} \text { satisfies } \mathscr{P}] \leq q^{-\left(\varepsilon-o_{n \rightarrow \infty}(1)\right) n}.
$$
\end{theorem}

This theorem offers a technique for constructing random-like codes from low biased codes. By combining puncturing and shortening methods, we can transform a code with weaker initial conditions into a random-like code. Since a monotone-increasing property is the negation of a monotone-decreasing property, we can present our theorem using the monotone-decreasing property and the term ``with high probability" instead of ``with probability exponentially decreasing".

\begin{theorem}
Let $0 < \varepsilon < 1$, $b \in \mathbb{N}$, and $q$ be a prime power. For any monotone-decreasing, $b$-local, and row-symmetric property $\mathscr{P}$ over $\mathbb{F}_q^n$ with $\mathrm{RLC}(\mathscr{P}) >\varepsilon$, there exists some $\eta > 0$ such that the following holds. If any one of the following conditions is satisfied for $R,\delta,\delta^{\perp},q,\eta$:
\begin{enumerate}
\item $\delta > (\frac{q - 1}{q} - \eta)(1 - R)$, or
\item $0 < \delta^{\perp} < \min\{\eta^{\frac{1}{\gamma}},\left( \frac{1 + \log_q(1 - \delta)}{36}\right)^2,(\frac{1}{q})^{\frac{1}{\gamma}}\}$ and $0 < R < \frac{\frac{1}{2} - 2\gamma}{1 + 0.9\log_q(1 - \delta)}\mathrm{H}_q(\delta^{\perp})$ for some $\gamma$,
\end{enumerate}
then there exist positive $m,p,s$ such that for any $[m, Rm, \delta m]_q$ code $\mathcal{C}$, upon performing a random $sm$-shortening and then a random $pm$-puncturing to $\mathcal{C}$, the resultant code $\mathcal{D}$ has length $n$, rate of at least $\mathrm{RLC}(\mathscr{P}) - \varepsilon$ and, with high probability, satisfies $\mathscr{P}$.
\end{theorem}

\begin{proof}
Define $ \eta=\min \{\frac{\varepsilon b \ln q}{q^b}, \mathrm{RLC}(\mathscr{P}) - \varepsilon \}$. Given that one of the two conditions holds, we can select an $s$ such that a random $sn$-puncturing to $\mathcal{C}$ results in an $\eta$-biased code by Theorem \ref{thm_main12} and Theorem \ref{thm_main2}. In this scenario, due to the singleton bound, the code rate is less than $\mathrm{RLC}(\mathscr{P}) - \varepsilon$. Utilizing Theorem \ref{thm:GM22}, we designate the rate as $\mathrm{RLC}(\mathscr{P}) - \varepsilon $ and perform a random $p$-puncturing. The resulting code then has a rate $\mathrm{RLC}(\mathscr{P}) - \varepsilon$ and, with a high probability, satisfies $\mathscr{P}$. Furthermore, as the choices of $s$ and $p$ are independent of $m$, we can finally define $m =\frac{n}{1-s-p}$. Consequently, the resulting code has a length of $n$.
\end{proof}

\section{Estimation on low-biased codewords}\label{sec_est}

For a random vector $x\in\mathbb{F}_q^n$, it is known from the law of large numbers that its empirical distribution $\mathrm{Emp}_x$ is, with high probability, $\varepsilon$-close to the uniform distribution over $\mathbb{F}_q$ for any $\varepsilon$ as $n$ goes to infinity. Therefore, for each $\varepsilon$, let $\mathcal{C}$ be a random code; $\mathcal{C}_\varepsilon$ will, with high probability, constitute only a small fraction of $\mathcal{C}$. In the following, we present several estimation methods for the size of $|\mathcal{C}_\varepsilon|$ under general conditions.

We first give a fact about field trace map: $\mathbb{F}_q \rightarrow \mathbb{F}_p$. The trace map is defined by $\mathrm{tr}(x) = \sum_{i=1}^{r} x^{p^i}$, where $q = p^r$. In fact, this map is both linear and surjective. Consequently, for any $a \in \mathbb{F}_q^{\ast}$, if $x$ is selected uniformly at random from $\mathbb{F}_q$, then $\mathrm{tr}(a \cdot x_i)$ will also assume a uniformly random value in $\mathbb{F}_p$.

We commence by presenting our first estimation method. Given a scenario where the distance $\delta$ is substantial, it is viable to set $\varepsilon$ to a small value, thereby ensuring that the number of codewords in $\mathcal{C}_{\varepsilon}$ is restricted to a maximum polynomial number. To exemplify, consider a binary code $\mathcal{C}$ over $\mathbb{F}_2^n$ with a considerable distance $\delta$. We argue that, in this instance, only a polynomial number of codewords surpass a weight greater than $\frac{1}{2}+O(\sqrt{\epsilon})$. We can create a Hamming sphere, centered at $[1,1,\cdots,1]$ with a radius $J(\delta)= \frac{1}{2}(1-\sqrt{1-2\delta})$. It is within this sphere that only a polynomial number of codewords exist.
\begin{lemma}\label{lem_upper1}
    Let $\mathcal{C}$ be an $[n,Rn,\delta n]_q$ code. For any $\varepsilon \geq 2(q-1) \sqrt{\frac{q-1}{q}(\frac{q-1}{q}-\delta)}$, $|\mathcal{C}_\varepsilon|\leq q^2 \delta n^2$.
\end{lemma}

\begin{proof}
    For each $t\in\mathbb{F}_q$ Denote $\vec{t}\in\mathbb{F}_q^n$ to be the vector where each entry has a value $t$. Denote $J_q(\delta)=(1-\frac{1}{q})\left(1-\sqrt{1-\frac{q}{q-1}\delta}\right)$, which is the list decoding radius associated with the Johnson bound. Let
    \begin{equation}
    B(t)=\{c\in C \mid d(c,\vec{t})\leq J_\delta(\delta)\}.
    \end{equation}
    By Johnson bound for a linear code,
    \begin{equation}
    |\bigcup_{t\in\mathbb{F}_q} B(t)|\leq \sum_{t\in\mathbb{F}_q} |B(t)|\leq  q^2 \delta n^2.
    \end{equation}
    If $c\notin\bigcup_{t\in\mathbb{F}_q} B(t)$, then the Hamming distance between $c$ and $\vec{t}$ is at least $J_q(\delta)$ for each $t\in \mathbb{F}_q$, which means 
    \begin{equation}
        \mathrm{Emp}_c(t) \leq 1 - J_q(\delta) = \frac{1}{q} +\sqrt{\frac{q-1}{q}(\frac{q-1}{q}-\delta)}.
    \end{equation}
    By Lemma \ref{lem_uniform}, $c$ is $\varepsilon$-biased for any $\varepsilon \geq 2(q-1) \sqrt{\frac{q-1}{q}(\frac{q-1}{q}-\delta)}$, and thus
    \begin{equation}
    |\mathcal{C}_\varepsilon|\leq |\bigcup_{t\in\mathbb{F}_q} B(t)|\leq q^2 \delta n^2.
    \end{equation}
\end{proof}

Another approach to approximate $|\mathcal{C}_{\varepsilon}|$ is the probability method. It is essential to observe that when the dual code of $\mathcal{C}$ has distance $d+1$, every set of $d$ columns within the generator matrix of $\mathcal{C}$ are linearly independent. This observation implies that when examining the distribution of a randomly selected codeword from $\mathcal{C}$, the bits exhibit $d$-wise independence. Consequently, $\mathcal{C}$ is bound by the constraints of the $d$-th moment inequality.

\begin{lemma}\label{lem_moment}
    $x_1, \cdots, x_n$ are independent random variables with $\mu=0$, and $x_i \in\left[-1, 1\right]$. Denote $X_n=\sum_{i=1}^n x_i$. Then for any even $d$,
    \begin{equation}
    \mathbb{E}((X_n)^d)  \leq 2 \cdot(2n)^{d/2} \cdot(\frac{d}{2})!.
    \end{equation}
\end{lemma}

\begin{proof}
By Hoeffding inequality, for any $\lambda \geq 0$,
\begin{equation}
\mathbb{E} e^{\lambda X_n} \leq e^{\frac{\lambda^2}{2} n}.
\end{equation}
Then by Chernoff bound, for any $t>0$,
\begin{equation}
\begin{gathered}
\mathrm{Pr}\left(X_n \geq t\right) \leq \exp \left(-\frac{t^2}{2n}\right); \\
\mathrm{Pr}\left(X_n \leq-t\right) \leq \exp \left(-\frac{t^2}{2n}\right).
\end{gathered}
\end{equation}
Then by Sub-Gaussian property, for any even $d$,
\begin{equation}
\mathbb{E}((X_n)^d) = \mathbb{E}(|X_n|^d) \leq 2 \cdot(2n)^{d/2} \cdot(\frac{d}{2})!.
\end{equation}
\end{proof}

\begin{corollary}\label{cor_tailupper}
    Let $x_1,x_2\cdots,x_n$ be random variables taking values in $[-1,1]$ which are $d$-wise independent, $\mathbb{E}(x_i)=0$. Let $X_n=\sum_{i=1}^n x_i$, $\delta=d/n$, then for any $\varepsilon>0$,
    \begin{equation}
    \mathrm{Pr}(|\sum_{i=1}^n x_i|\geq\varepsilon n)\leq 4\sqrt{\pi d}(\frac{\delta}{\varepsilon^2 e})^{\delta n/2}.
    \end{equation}
\end{corollary}

\begin{proof}
    If $x_1,x_2\cdots,x_n$ are $d$-wise independent, then the $d$-th central moment of $X_n=\sum_{i=1}^n x_i$ is the same as the case where these $x_1,x_2\cdots,x_n$ are fully independent. Therefore, by Markov inequality,
    \begin{equation}
        \mathrm{Pr}(|\sum_{i=1}^n x_i|\geq\varepsilon n)\leq \frac{\mathbb{E}((X_n)^d) }{(\varepsilon n)^2} =4\sqrt{\pi d}(\frac{\delta}{\varepsilon^2 e})^{\delta n/2}.
    \end{equation}
\end{proof}

\begin{lemma}\label{lem_probupper}
    Let $x$ be a random vector, whose components uniformly take values in $\mathbb{F}_q$ and are $d$-wise independent. Let $\delta=d/n$. Then 
    \begin{equation}
    \mathrm{Pr}(x \text{\ is not $\varepsilon$-biased})\leq 2\sqrt{2}(q-1)\left( \frac{2\delta}{\varepsilon^2 e}\right)^{\delta n/2}.
\end{equation}
\end{lemma}

\begin{proof}
    Let $\mathrm{Re}(\cdot)$ and $\mathrm{Im}(\cdot)$ denote the real part and imaginary part of a complex number separately. Since for each $i\in[n]$, $a\in\mathbb{F}_q$, $\mathrm{Re}(\omega^{\mathrm{tr}(a\cdot x_i)})$ and $\mathrm{Im}(\omega^{\mathrm{tr}(a\cdot x_i)})$ are real-valued discrete random variables with $\mu=0$, and taking values in $[0,1]$.
    \begin{equation}
        \begin{aligned}
    &\mathrm{Pr}(x \text{\ is not $\varepsilon$-biased})\\
    \leq&\sum_{a\in \mathbb{F}_q^{\ast}}\mathrm{Pr}(|\sum_{i=1}^{n}\omega^{\mathrm{tr}(a\cdot x_i)}|\geq \varepsilon n)\\
    \leq&\sum_{a\in \mathbb{F}_q^{\ast}}\left(\mathrm{Pr}(|\sum_{i=1}^{n}\mathrm{Re}(\omega^{\mathrm{tr}(a\cdot x_i)})|\geq \frac{\sqrt{2}}{2}\varepsilon n)+\mathrm{Pr}(|\sum_{i=1}^{n}\mathrm{Im}(\omega^{\mathrm{tr}(a\cdot x_i)})|\geq \frac{\sqrt{2}}{2}\varepsilon n)\right)\\
    \leq&8(q-1)\sqrt{\pi \delta n}\left( \frac{2\delta}{\varepsilon^2 e}\right)^{\delta n/2}.
    \end{aligned}
    \end{equation}
\end{proof}

\begin{corollary}\label{cor_upper2}
    Let $\mathcal{C}$ be a code of length $n$, rate $R$ and dual distance $d^{\perp}=\delta^{\perp}n$ over the field $\mathbb{F}_q$. Then for each $\varepsilon>0$, the number of codewords which are not $\varepsilon$-biased is not more than
    $$
    8q\sqrt{\pi \delta^{\perp} n}\cdot \left( \frac{2\delta^{\perp}}{\varepsilon^2 e}\right)^{\delta^{\perp} n/2}\cdot q^{Rn}
    $$
    for sufficiently large $n$.
\end{corollary}
\begin{proof}
If $\mathcal{C}$ has dual distance $d^{\perp}$, then each $d^{\perp}-1$ columns of generating matrix of $\mathcal{C}$ are independent, which means that if we uniformly randomly choose a codeword from $\mathcal{C}$, each $d^{\perp}-1$ bits of the codeword are independent. Therefore, the number of codewords is less than 
\begin{equation}
    \begin{aligned}
    &\mathrm{Pr}\left(\text{$c$ is not $\varepsilon$-biased} \mid c\in\mathcal{C}\right)\cdot|\mathcal{C}|\\
    \leq & 8(q-1)\sqrt{\pi (\delta^{\perp} n-1)}\cdot \left( \frac{2\delta^{\perp}-\frac{2}{n}}{\varepsilon^2 e}\right)^{(\delta^{\perp} n-1)/2}\cdot q^{Rn},
    \end{aligned}
\end{equation}
which is less than 
$$8q\sqrt{\pi \delta^{\perp} n}\cdot \left( \frac{2\delta^{\perp}}{\varepsilon^2 e}\right)^{\delta^{\perp} n/2}\cdot q^{Rn}$$
for sufficiently large $n$.
\end{proof}

\section{Proof of Theorem \ref{thm_main1}}\label{sec_prove1}

Before proving Theorem \ref{thm_main1}, we first give the following theorem.

\begin{theorem}\label{thm_main1ex}
    Let $\mathcal{C}$ be an $[n,Rn,\delta n]_q$ code. If we perform a random $sn$-shortening $\mathcal{S}$ to $\mathcal{C}$, where $s < R$, then with high probability, the shortened code $\mathcal{C}^{[\mathcal{S}]}$ is $\varepsilon$-biased, where $\varepsilon = \frac{2(q - 1) \sqrt{\frac{q - 1}{q}\left(\frac{q - 1}{q} - \delta\right)} + s}{1 - s}$.
\end{theorem}
    
\begin{proof}
    Denote $\varepsilon^{\prime} = 2(q-1) \sqrt{\frac{q-1}{q}(\frac{q-1}{q}-\delta)}$. From Lemma \ref{lem_upper1}, we know that $|\mathcal{C}_{\varepsilon^{\prime}}|\leq q^2\delta n^2$. Hence, by union bound,
\begin{equation}
    \begin{aligned}
    &\mathrm{Pr}\left(\text{$\mathcal{S}$ doesn't hit all codewords in $\mathcal{C}_{\varepsilon^{\prime}}$}\right)\\
    \leq& \sum_{c\in \mathcal{C}_{\varepsilon^{\prime}}} \mathrm{Pr}\left(\text{$c$ is not hit by $\mathcal{S}$}\right)\\
    \leq& q^2\delta n^2\cdot (1-\delta)^{sn},
    \end{aligned}
\end{equation}
    which approaches $0$ as $n$ tends to infinity. Hence, with high probability, $\mathcal{S}$ hits $\mathcal{C}_{\varepsilon^{\prime}}$. From Lemma \ref{lem_eps}, the shortened code $\mathcal{C}^{[\mathcal{S}]}$ is $\varepsilon$-biased.
\end{proof}
To get Theorem \ref{thm_main1}, we need to select an appropriate value for $s$ that satisfies the following conditions:
\begin{enumerate}
    \item The chosen $s$ does not cause a significant decrease in the rate $R$.
    \item $s$ satisfies the requirements defined by $\varepsilon$.
\end{enumerate}

\begin{proof}[Proof of Theorem \ref{thm_main1}]
    Let $s=\min\{\frac{\gamma}{1+\gamma},\frac{R}{2},\frac{\varepsilon}{2}-(q-1)\sqrt{\frac{q-1}{q}(\frac{q-1}{q}-\delta)}\}$. By Proposition \ref{pro_rate}, since $s\leq\min\{\frac{\gamma}{1+\gamma},\frac{R}{2}\}$, the rate of $\mathcal{C}^{[\mathcal{S}]}$ is 
    \begin{equation}
    \frac{R-s}{1-s}>R-\gamma.
    \end{equation}
    Rearranging the inequality $\frac{q-1}{q}-\frac{q}{q-1}\left(\frac{\varepsilon}{2(q-1)}\right)^2<\delta$, we get
    \begin{equation}
        (q-1)\sqrt{\frac{q-1}{q}(\frac{q-1}{q}-\delta)}<\frac{\varepsilon}{2},
    \end{equation} 
    And since $s<\frac{\varepsilon}{2}-(q-1)\sqrt{\frac{q-1}{q}(\frac{q-1}{q}-\delta)}$.
    \begin{equation}
    \begin{aligned}
    &\frac{2(q-1) \sqrt{\frac{q-1}{q}(\frac{q-1}{q}-\delta)}+s}{1-s} \\ <& \frac{(q-1) \sqrt{\frac{q-1}{q}(\frac{q-1}{q}-\delta)}+\frac{\varepsilon}{2}}{1-\frac{\varepsilon}{2}+(q-1) \sqrt{\frac{q-1}{q}(\frac{q-1}{q}-\delta)}}\\<& \varepsilon.
    \end{aligned}
    \end{equation}
    By Lemma \ref{thm_main1ex}, $\mathcal{C}^{[\mathcal{S}]}$ is $\varepsilon$-biased.
\end{proof}


\begin{proof}[Proof of Theorem \ref{thm_main12}]
    Since $\frac{\delta}{1 - (1-\beta)R} > \frac{q - 1}{q} - \frac{q}{q - 1}\left(\frac{\varepsilon}{2(q - 1)}\right)^2$, we can first choose an $s<(1-\beta)R$ such that $\frac{\delta}{1 - s} > \frac{q - 1}{q} - \frac{q}{q - 1}\left(\frac{\varepsilon}{2(q - 1)}\right)^2$.
    We then perform an $sn$-shortening $\mathcal{S}$ to $\mathcal{C}$. By Proposition \ref{pro_dist} and Proposition \ref{pro_rate}, $\mathcal{C}^{[\mathcal{S}]}$ has a distance of at least $\frac{q - 1}{q} - \frac{q}{q - 1}\left(\frac{\varepsilon}{2(q - 1)}\right)^2$ and rate at least $\frac{R-s}{1-s}>\beta R$. Then, using Theorem \ref{thm_main1}, we are able to select a sufficiently small $\gamma$ defined in Theorem \ref{thm_main1} such that $\frac{R-s}{1-s} - \gamma > \beta R$, which ultimately enables us to achieve the desired result.
\end{proof}
\section{Proof of Theorem \ref{thm_main2}}\label{sec_pro2}

We first present an inequality concerning the $q$-ary entropy function $\mathrm{H}_q(x)$ here.

\begin{lemma}\label{lemma_Hinq}
    For any $q=p^r$, $0<\gamma<\frac{1}{4}$, when $0<x<(\frac{1}{q})^{\frac{1}{\gamma}}$, $\mathrm{H}_q(x)<-(1+2\gamma)x\log_q x$.
\end{lemma}
\begin{proof}
When $q=2$,
\begin{equation}
    \frac{\mathrm{H}_q(x)}{-x\log_q x} -1 = \frac{-(1-x)\log_q(1-x)}{-x\log_q(x)} < \frac{\log_q(1+2x)}{-x\log_q(x)} < \frac{2}{-\log_q(x)} =-2\log_{x}(q)<2\gamma.
\end{equation} 
for $0<x<(\frac{1}{q})^{\frac{1}{\gamma}}$. When $q\geq 3$,
\begin{equation}
\begin{gathered}
    \frac{x\log_q(q-1)}{-x\log_q(x)} = -\log_{x}(q-1) < -\log_{x}(q) < \gamma,\\
    \frac{-(1-x)\log_q(1-x)}{-x\log_q(x)} < \frac{-\log_q(1-x)}{-x\log_q(x)} < \frac{x}{-x\log_q(x)} = -\log_{x}(q)< \gamma,
\end{gathered}
\end{equation}
for $0<x<(\frac{1}{q})^{\frac{1}{\gamma}}$, and thus
\begin{equation}
    \frac{\mathrm{H}_q(x)}{-x\log_q x} -1=\frac{x\log_q(q-1)}{-x\log_q(x)}+\frac{-(1-x)\log_q(1-x)}{-x\log_q(x)} < 2\gamma.
\end{equation} 
\end{proof}

\begin{proof}[Proof of Theorem \ref{thm_main2}]
    We set $\varepsilon^{\prime}=0.9\varepsilon$ and get $\delta_0^{\perp}<(\frac{\sqrt{e}\varepsilon^{\prime}}{\sqrt{2}})^{\frac{1}{\gamma}}$. Let $s=\frac{R-(\frac{1}{2}-2\gamma)\mathrm{H}_q(\delta_0^{\perp})}{-\log_q(1-\delta)}$. We get
\begin{equation}\label{union2}
    \begin{aligned}
&\mathrm{Pr}\left(\text{$\mathcal{S}$ doesn't hit all codewords in $\mathcal{C}_{\varepsilon^{\prime}}$}\right)\\
    \leq& \sum_{c\in \mathcal{C}_{\varepsilon^{\prime}}} \mathrm{Pr}\left(\text{$c$ is not hit by $\mathcal{S}$}\right)\\
    \leq & |\mathcal{C}_{\varepsilon^{\prime}}|\cdot (1-\delta)^{sn}&&(\text{using Lemma \ref{lem_hitprob}})\\
    \leq& 8q\sqrt{\pi \delta_0^{\perp} n}\cdot \left( \frac{2\delta_0^{\perp}}{(\varepsilon^{\prime})^2 e}\right)^{\delta_0^{\perp} n/2}\cdot q^{Rn} \cdot (1-\delta)^{sn}&&(\text{using Corollary \ref{cor_upper2}})\\
    =& 8q\sqrt{\pi \delta_0^{\perp} n}\cdot \left( \frac{2\delta_0^{\perp}}{(\varepsilon^{\prime})^2 e}\right)^{\delta_0^{\perp} n/2}\cdot q^{Rn} \cdot q^{\log_q(1-\delta)sn}\\
    = & 8q\sqrt{\pi \delta_0^{\perp} n}\cdot\left((\delta_0^{\perp})^{1-2\gamma} \right)^{\delta_0^{\perp} n/2}\cdot q^{(\frac{1}{2}-2\gamma) \mathrm{H}_q(\delta_0^{\perp}) \cdot n}&&(\text{using } s=\frac{R-(\frac{1}{2}-2\gamma)\mathrm{H}_q(\delta_0^{\perp})}{-\log_q(1-\delta)})\\
    \leq & 8q\sqrt{\pi \delta_0^{\perp} n}\cdot\left((\delta_0^{\perp})^{1-2\gamma} \right)^{\delta_0^{\perp} n/2}\cdot q^{-(\frac{1}{2}-2\gamma) (1+2\gamma)\delta_0^{\perp}\log_q\delta_0^{\perp} \cdot n}&&(\text{using Lemma \ref{lemma_Hinq}})\\
    = & 8q\sqrt{\pi \delta_0^{\perp} n}\cdot\left((\delta_0^{\perp})^{\frac{1}{2}-\gamma}\cdot (\delta_0^{\perp})^{-(\frac{1}{2}-2\gamma) (1+2\gamma)}\right)^{\delta_0^{\perp}n}\\
    = & 8q\sqrt{\pi \delta_0^{\perp} n}\cdot(\delta_0^{\perp})^{4\gamma^2\delta_0^{\perp} n}.
    \end{aligned}
\end{equation}

Therefore, this probability in Equation \ref{union2} tends to $0$ as $n$ approaches infinity. 

Moreover, one can verify that $s<0.9 R$ given $s=\frac{R-(\frac{1}{2}-2\gamma)\mathrm{H}_q(\delta_0^{\perp})}{-\log_q(1-\delta)}$ and $R\le \frac{\frac{1}{2}-2\gamma}{1+0.9\cdot\log_q(1-\delta)}\mathrm{H}_q(\delta_0^{\perp})$. Hence, the shortened code $\mathcal{C}^{[\mathcal{S}]}$ has rate at least $\frac{R-s}{1-s}>0.1 R$, and
\begin{equation}
    \begin{aligned}
    s &< R \\
      &< \frac{0.5-2\gamma}{1+\log_q(1-\delta)}\mathrm{H}_q(\delta_0^{\perp}) \\
      &< -\frac{0.75}{1+\log_q(1-\delta)} \cdot\delta_0^{\perp}\cdot\log_q(\delta_0^{\perp}) \quad\quad\quad\quad\quad\quad\quad\quad\quad\quad(\text{using Lemma \ref{lemma_Hinq}}) \\
      &< \frac{0.75}{1+\log_q(1-\delta)} \cdot(\delta_0^{\perp})^{\frac{1}{2}}\cdot(\delta_0^{\perp})^{\gamma}\cdot(\delta_0^{\perp})^{\frac{1}{4}}\cdot(-\log_q(\delta_0^{\perp})) \quad\quad(\text{using } \gamma<\frac{1}{4}) \\
      &< \frac{1}{48}\cdot \frac{10}{9}\varepsilon^{\prime} \cdot\frac{4}{e\ln(2)} \\
      &\quad(\text{using } 0<\delta_0^{\perp}<\min\{\left( \frac{2+\log_q(1-\delta)}{36}\right)^2, \varepsilon^{\frac{1}{\gamma}}\} \text{ and }-(\delta_0^{\perp})^{\frac{1}{4}}\cdot\log_q(\delta_0^{\perp}) \leq \frac{4}{e\ln2}) \\
      &< 0.05\varepsilon^{\prime}.
    \end{aligned}
\end{equation}

By Lemma \ref{lem_eps}, $\mathcal{C}^{[\mathcal{S}]}$ is $\frac{\varepsilon^{\prime}+s}{1-s}$-biased and since $\frac{\varepsilon^{\prime}+s}{1-s}< \frac{10}{9}\varepsilon^{\prime}<\varepsilon$. So, with high probability, $\mathcal{C}^{[\mathcal{S}]}$ is $\varepsilon$-biased.
\end{proof}

\begin{proof}[Proof of Corollary \ref{cor_main2cor}]
Given that $\delta > 1 - q^{-0.6}$, we have $\log_q(1 - \delta) < -0.6$. We can select a universal constant $\eta$ such that $\frac{\frac{1}{2} - 2\eta}{1 + 0.9 \cdot \log_q(1 - \delta)} > 1$. Set $\gamma = \min\{ \eta , \frac{\frac{1}{2} - 2\eta}{1 + 0.9 \cdot \log_q(1 - \delta)} - 1\}$. By invoking theorem \ref{thm_main2} with $s=\frac{R-(\frac{1}{2}-2\gamma)\mathrm{H}_q(\delta_0^{\perp})}{-\log_q(1-\delta)}$ again, we have the same guarantee.
 \end{proof}